\documentclass[a4paper,UKenglish]{lipics-v2018}
\usepackage{microtype}
\usepackage[ruled,vlined]{algorithm2e}

\title{Conversion from RLBWT to LZ77}
\titlerunning{Conversion from RLBWT to LZ77}
\author{Takaaki Nishimoto}{RIKEN Center for Advanced Intelligence Project, Japan}{takaaki.nishimoto@riken.jp}{}{}
\author{Yasuo Tabei}{RIKEN Center for Advanced Intelligence Project, Japan}{yasuo.tabei@riken.jp}{}{}
\authorrunning{T. Nishimoto and Y. Tabei}
\Copyright{T. Nishimoto and Y. Tabei}
\subjclass{Theory of computation $\rightarrow$ Design and analysis of algorithms $\rightarrow$ Data structures design and analysis}
\category{}
\relatedversion{}
\supplement{}
\funding{}
\acknowledgements{}

\EventEditors{John Q. Open and Joan R. Access}
\EventNoEds{2}
\EventLongTitle{42nd Conference on Very Important Topics (CVIT 2016)}
\EventShortTitle{CVIT 2016}
\EventAcronym{CVIT}
\EventYear{2016}
\EventDate{December 24--27, 2016}
\EventLocation{Little Whinging, United Kingdom}
\EventLogo{}
\SeriesVolume{42}
\ArticleNo{23}
\nolinenumbers 
\hideLIPIcs  

\begin{document}

\maketitle

\begin{abstract}
Converting a compressed format of a string into another compressed format without an explicit decompression is one of the central research topics in string processing. 
We discuss the problem of converting the run-length Burrows-Wheeler Transform~(RLBWT) of a string to Lempel-Ziv 77~(LZ77) phrases of the reversed string.
The first results with Policriti and Prezza’s conversion algorithm~[Algorithmica 2018] 
were $O(n \log r)$ time and $O(r)$ working space 
for length of the string $n$, number of runs $r$ in the RLBWT, and number of LZ77 phrases $z$.
Recent results with Kempa’s conversion algorithm~[SODA 2019] are $O(n / \log n + r \log^{9} n + z \log^{9} n)$ time 
and $O(n / \log_{\sigma} n + r \log^{8} n)$ working space for the alphabet size $\sigma$ of the RLBWT.
In this paper, we present a new conversion algorithm by improving Policriti and Prezza's conversion algorithm where dynamic data structures for general purpose are used. 
We argue that these dynamic data structures 
can be replaced and present new data structures for faster conversion. 
The time and working space of our conversion algorithm 
with new data structures are 
$O(n \min \{ \log \log n, \sqrt{\frac{\log r}{\log\log r}} \})$ and $O(r)$, respectively. 

\end{abstract}
\section{Introduction}\label{intro}
Converting a compressed format of a string into another compressed format without an explicit decompression is one of the central research topics in string processing.
Examples are conversions from the Lempel-Ziv 77~(LZ77) Phrases of a string into a grammar-based encoding~\cite{DBLP:journals/tcs/Jez16,DBLP:journals/tcs/Rytter03}, from 
a grammar-based encoding of a string into LZ78 phrases~\cite{DBLP:conf/spire/BannaiIT12,DBLP:conf/cpm/BannaiGIT13} and 
from a grammar-based encoding of a string into another grammar-based encoding~\cite{DBLP:journals/corr/abs-1811-01472}.
Such conversion is beneficial when one intends to process a compressed string to a different compressed format without decompressing it. 

{\em LZ77 parsing}, proposed in 1976~\cite{LZ76}, is one of the most popular lossless data compression algorithms and is 
a greedy partition of a string such that each phrase is a previous occurrence of a substring or a character not occurring previously in the string.
The {\em run-length Burrows-Wheeler transform (RLBWT)}~\cite{burrows1994block} is a recent popular lossless data compression algorithm 
with a run-length encoded permutation of a string. 

Policriti and Prezza~\cite{DBLP:journals/algorithmica/PolicritiP18} proposed the first conversion algorithm from the RLBWT 
of an input string to the LZ77 phrases of the reversed string.
The basic idea with this algorithm is to carry out backward searches on the RLBWT and find a previous occurrence of each phrase 
using red-black trees storing a sampled suffix array and dynamic data structure for solving 
\emph{the searchable partial sums with the indels problem}~(e.g., \cite{DBLP:journals/algorithmica/BilleCCGSVV18,DBLP:journals/tcs/HonSS11}). 
Since these data structures are updated frequently for scanning the string in the RLBWT format, 
the running time and working space are $O(n \log r)$ and $O(r)$ words, respectively, for string length $n$ and number of runs $r$ (i.e., the number of continuous occurrences of the same characters). 

Kempa~\cite{DBLP:conf/soda/Kempa19} recently presented a conversion algorithm from the RLBWT to LZ77, which runs in $O(n / \log n + r \log^{9} n + z \log^{9} n)$ time and 
$O(n / \log_{\sigma} n + r \log^{8} n)$ working space for the number $z$ of LZ77 phrases 
and the alphabet size $\sigma$ of the string in the RLBWT format. 
While the algorithm runs in $o(n)$ time and working space, especially when $r$ and $z$ are small~(e.g., $r, z = O(n / \log^{9} n)$), 
the working space of the algorithm is larger than that of Policriti and Prezza's algorithm in many cases. 

In this paper, we present a new conversion algorithm from the RLBWT to LZ77 by 
improving Policriti and Prezza's algorithm. 
Their algorithm adopts dynamic data structures for four queries consisting of backward search, LF function, access queries on the RLBWT, 
and so-called {\em range more than query (RMTQ)}.
We argue that these dynamic data structures can be replaced for answering those queries and present new data structures for faster conversion. 
Our algorithm runs in $O(n \min \{ \log \log n, \sqrt{\frac{\log r}{\log\log r}} \})$ time 
and $O(r)$ working space, which is faster than their algorithm while using the same working space (see Table~\ref{table:result} for a summary of conversion algorithms). 
 
\begin{table}[t]
    \vspace{-0.5cm}
    \caption{Summary of conversion algorithms from the RLBWT to LZ77.}
    \label{table:result} 
    \center{	

    \begin{tabular}{r||c|c}
 Algorithm & Conversion time & Working space~(words) \\ \hline
Policriti and Prezza~(Thm. 7)~\cite{DBLP:journals/algorithmica/PolicritiP18} & $O(n \log r)$ & $O(r)$ \\ \hline
Kempa~(Thm. 7.3)~\cite{DBLP:conf/soda/Kempa19} & $O(n / \log n + r \log^{9} n + z \log^{9} n)$ & $O(n / \log_{\sigma} n + r \log^{8} n)$ \\ \hline\hline
This study & $O(n \min \{ \log \log n, \sqrt{\frac{\log r}{\log\log r}} \})$ & $O(r)$ \\ 
    \end{tabular} 
    }
\end{table}

\newcommand{\Occ}{\mathit{Occ}}

\newcommand{\floor}[1]{\left \lfloor #1 \right \rfloor}
\newcommand{\ceil}[1]{\left \lceil #1 \right \rceil}
\newcommand{\LZ}{\mathsf{LZ}}

\newcommand{\SE}[2]{\mathit{SE}_{#1}^{#2}}
\newcommand{\encblock}{\mathit{LC}}
\newcommand{\encpow}[1]{\mathit{RLF}(#1)}
\newcommand{\assign}{\mathit{Assgn}}
\newcommand{\argmax}{\mathop{\rm arg~max}\limits}
\newcommand{\argmin}{\mathop{\rm arg~min}\limits}

\newcommand{\rmq}{\mathsf{RMQ}}
\newcommand{\SA}{\mathsf{SA}}
\newcommand{\ISA}{\mathsf{ISA}}
\newcommand{\LCE}{\mathsf{LCE}}
\newcommand{\LF}{\mathsf{LF}}
\newcommand{\cfunc}{\mathsf{C}}
\newcommand{\rle}{\mathsf{RLE}}
\newcommand{\run}{\mathsf{run}}
\newcommand{\offset}{\mathsf{offset}}
\newcommand{\bsearch}{\mathsf{backward\_search}}
\newcommand{\any}{\mathsf{RMTQ}}

\newcommand{\RLZ}{\mathsf{RLZ}}
\newcommand{\fopen}{\mathsf{open}}
\newcommand{\fclose}{\mathsf{close}}

\newcommand{\rank}{\mathsf{rank}}
\newcommand{\select}{\mathsf{select}}
\newcommand{\access}{\mathsf{access}}
\newcommand{\pred}{\mathsf{pred}}

\section{Preliminaries} \label{sec:preliminary}
Let $\Sigma$ be an ordered alphabet of size $\sigma$, 
$T$ be a string of length $n$ over $\Sigma$ and 
$|T|$ be the length of $T$. 
Let $T[i]$ be the $i$-th character of $T$ and 
$T[i..j]$ be the substring of $T$ that begins at position $i$ and ends at position $j$. 
The $T[i..]$ denotes the suffix of $T$ beginning at position $i$, i.e., $T[i..n]$. 
Let $T^R$ be the reversed string of $T$, i.e., $T^R = T[n] T[n-1] \cdots T[1]$. 
For two integers $i$ and $j$~($i \leq j$), $[i, j]$ represents $\{i, i+1, \ldots, j \}$.
For two strings $T$ and $P$, $T \prec P$ is that $T$ is lexicographically smaller than $P$.
The $\Occ(T, P)$ denotes all the occurrence positions of string $P$ in string $T$, i.e., 
$\Occ(T, P) = \{i \mid P = T[i..i+|P|-1], i \in [1, n-|P|+1] \}$. 
\emph{Right occurrence} $p$ of substring $T[i..j]$ is 
a subsequent occurrence position of $T[i..j]$ in $T$, 
i.e., any $p \in \Occ(T[i+1..], T[i..j])$. 

For a string $T$, character $c$, and integer $i$, 
The $\rank(T, c, i)$ returns the number of a character $c$ in $T[..i]$, i.e., 
$\rank(T, c, i) = |\Occ(T[..i], c)|$. 
$\access(T, i)$ returns $T[i]$. 
The $\select(T, c, i)$ returns the position of the $i$-th occurrence of a character $c$ in $T$.
If the number of occurrences of $c$ in $T$ is smaller than $i$, it returns $n+1$, i.e., 
$\select(T, c, i) = \min (\{ j \mid |\Occ(T[..j], c)| \geq i, j \in [1, n] \} \cup \{ n + 1 \})$ 
where $\min \{S\}$ returns the minimum value in a given set $S$. 

For an integer $x$ and set $S$ of integers, 
a \emph{predecessor} query $\pred(S,x)$ returns 
the number of elements that are no more than $x$ in $S$, 
i.e., $\pred(S,x) = |\{ y \mid y \leq x, y \in S \}|$.
A \emph{predecessor data structure} of $S$ supports predecessor queries on $S$. 
For an integer array $D$ and two positions $i, j$~($i \leq j$) on $D$, 
a \emph{range maximum query}~(RMQ) $\rmq(D, i, j, k)$ returns 
the maximum value in $D[i..j]$, i.e., $\rmq(D, i, j) = \max D[i..j]$, 
where $\max\{S\}$ returns the maximum value in a given set $S$. 

Our computation model is a unit-cost word RAM with a machine word size of $\Omega(\log_2 n)$ bits. 
We evaluate the space complexity in terms of the number of machine words. A bitwise evaluation of 
space complexity can be obtained with a $\log_2 n$ multiplicative factor. 

\subsection{Suffix Array~(SA) and SA interval}
The suffix array~($\SA$)~\cite{DBLP:journals/siamcomp/ManberM93} of string $T$ is an integer array of size $n$ 
such that $\SA[i]$ stores the starting position of $i$-th suffix of $T$ in lexicographical order. 
Formally, $\SA$ is the permutation of $[1..n]$ such that $T[SA[1]..] \prec \cdots \prec T[SA[n]..]$ holds. 
For example, $T={\bf mississippi\$}$ and $\SA = 12,11,8,5,2,1,10,9,7,4,6,3$. 
The left figure in Figure~\ref{fig:sa} depicts the sorted suffixes of $T$ and $\SA$ of $T$. 

Since suffixes in the suffix array are sorted in the lexicographical order, 
suffixes with prefix $Y$ occur continuously on an interval in the suffix array. 
We call this interval the \emph{SA interval} of $Y$. 
Formally, the SA interval of string $Y$ is interval $[b, e]$ such that 
$p \in \SA[b..e]$ holds for all $p \in \Occ(T, Y)$. 
For the above example, the SA intervals of {\bf si} and {\bf p} are $[9, 10]$ and $[7, 8]$, 
respectively.

\begin{figure}[t]
\begin{center}
\begin{tabular}{ccc}

		\includegraphics[width=0.36\textwidth]{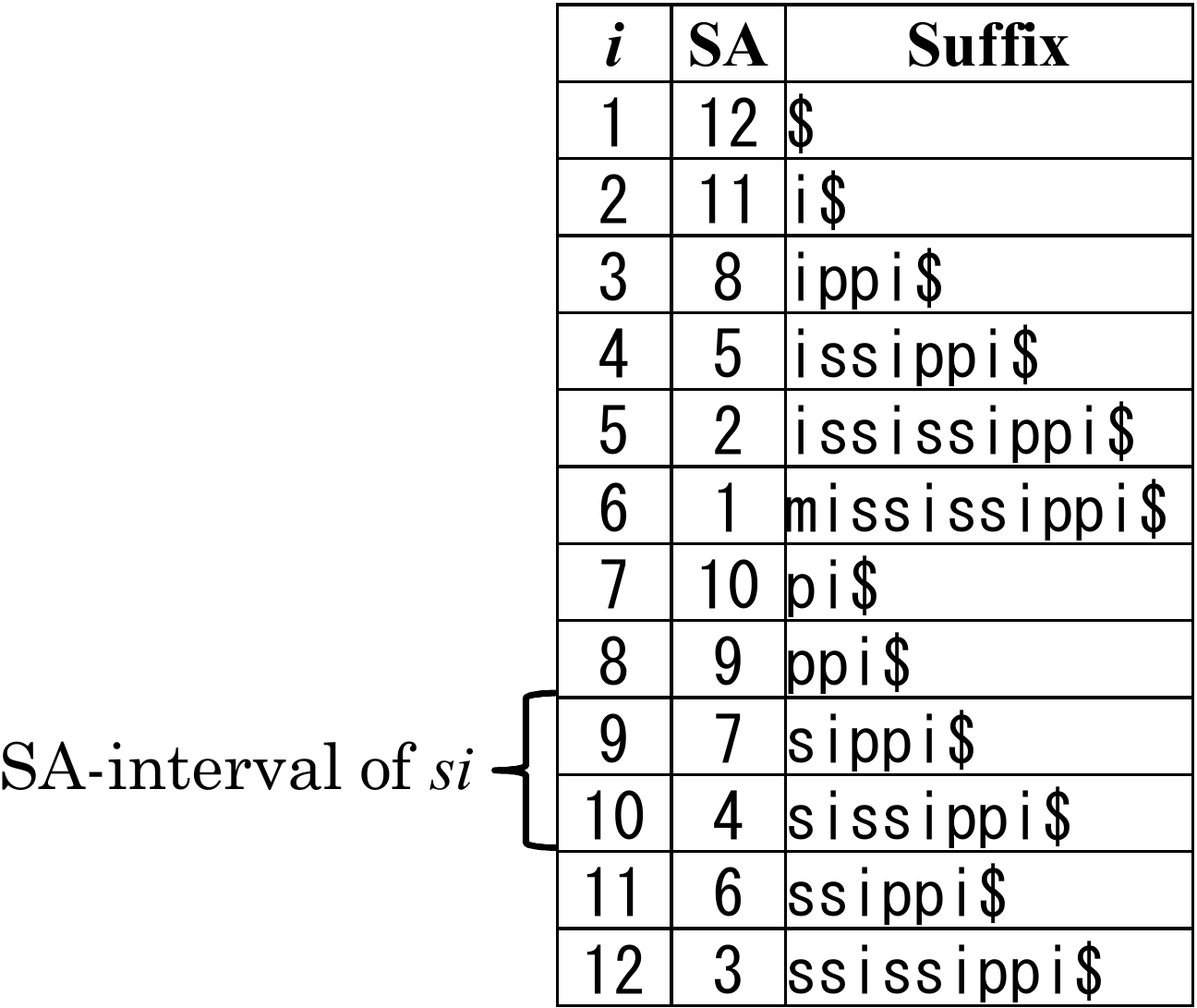} &
		\includegraphics[width=0.22\textwidth]{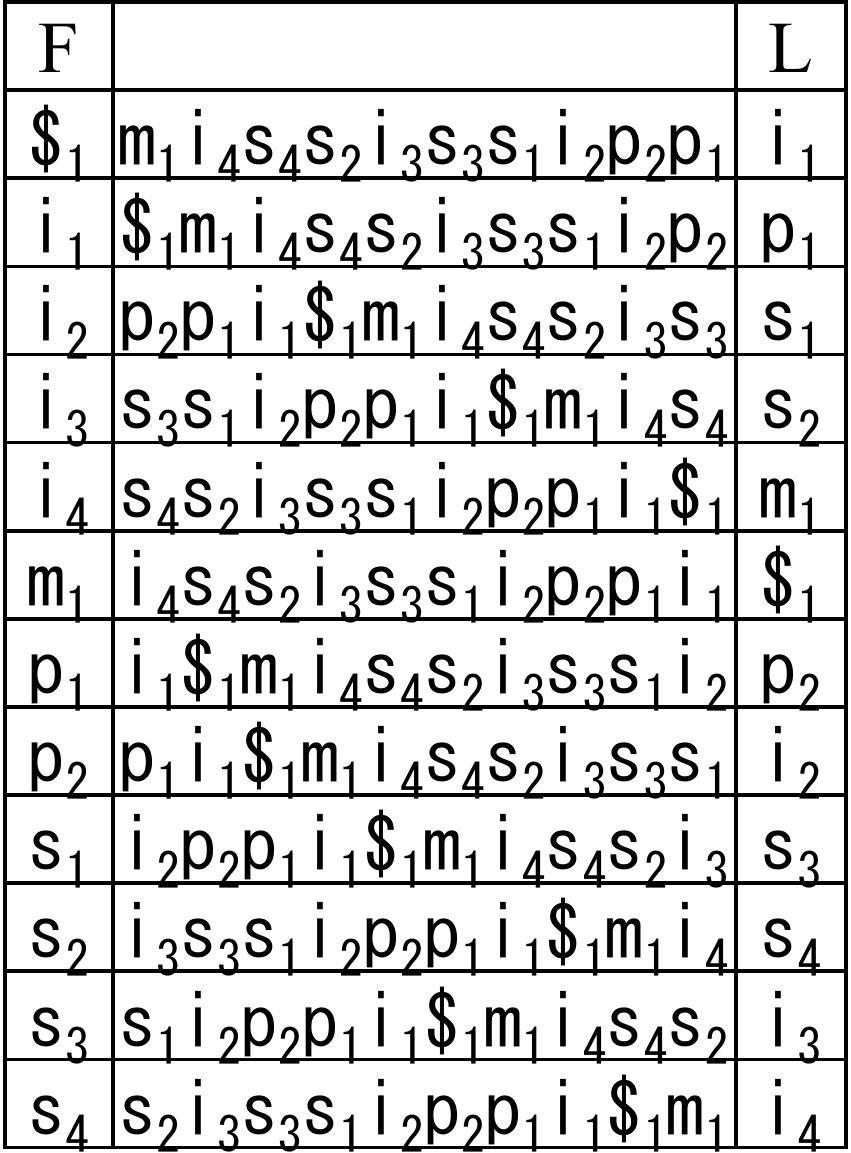} &
		\includegraphics[width=0.18\textwidth]{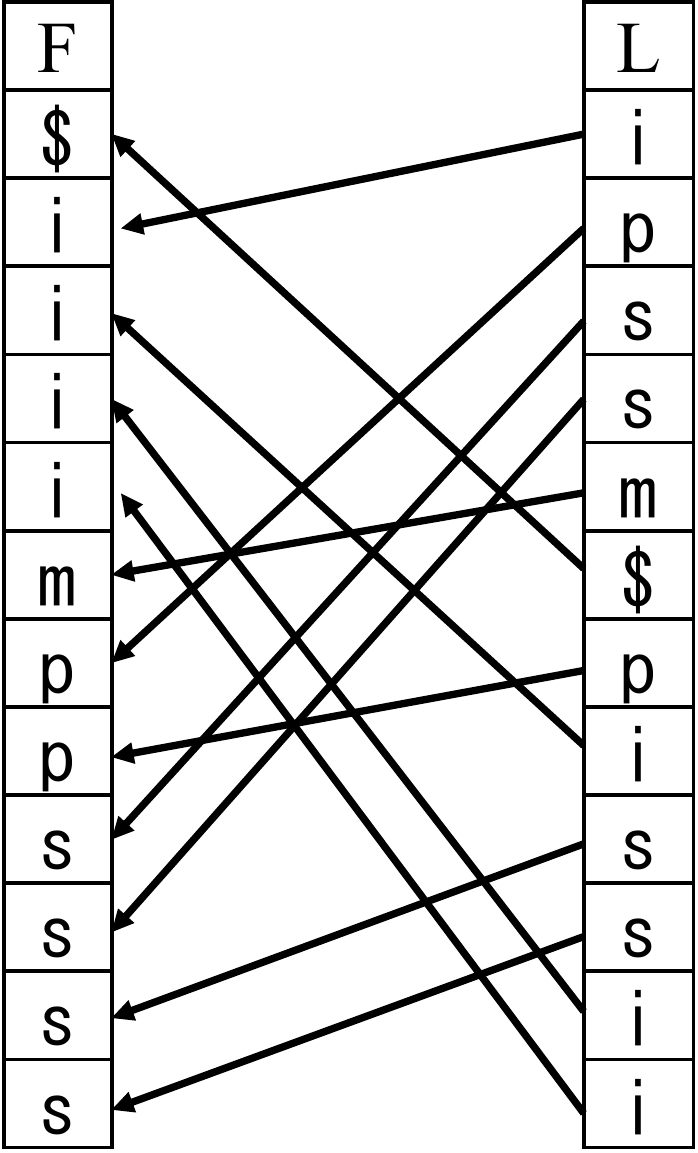} 
\end{tabular}
\end{center}
	  \caption{ 
	  Example for $\SA$ (left), $F$, $L$ (center), and LF function (right) of $T=mississippi\$$.
	  }
 \label{fig:sa}
\end{figure}

\subsection{BWT and backward search}
{\em The Burrows Wheeler Transform~(BWT)}~\cite{burrows1994block} of string $T$ is
a permutation of $T$ obtained as follows.
We sort all the $n$ rotations of $T$ in lexicographical order
and take the last character at each rotation in sorted order.
The $L$ is the permutation of $T$ such that for all $i \in [1..n]$, $L[i] = T[\SA[i]-1]$ holds if $\SA[i] \neq 1$ and $L[i] = T[n]$ holds otherwise.
Similarly, let $F$ be a permutation of $T$ that consists of the first characters in rotations in sorted order,
i.e., $F[i] = T[\SA[i]]$ for all $i \in [1..n]$.
The middle table in Figure~\ref{fig:sa} represents $F$, $L$, and sorted rotations of $T={\bf mississippi\$}$.

A property of BWT is that the $i$-th occurrence of character $c$ in $F$ corresponds to the $i$-th occurrence of $c$ in $L$. 
In other words, let $x$ and $y$ be the positions of the $i$-th occurrence of $c$ in $F$ and $L$, respectively.
Then $F[x]$ is a character of position $p$ in $T$ 
when $L[y]$ is a character of the same position in $T$. 
The $\LF$ function receives a position $y$ in $L$ as input 
and returns such corresponding position $x$ in $F$.
Since $F$ consists of the sorted characters, $\LF(y) = C[L[y]] + \rank(L, L[y], i)$ holds 
for all $y \in [1, n]$, where 
$C[c]$ is the number of occurrences of characters lexicographically less than $c\in \Sigma$ in $L$.

Backward search computes the SA interval of $cY$ for a given SA interval of $Y$
and character $c$ using the BWT of $T$.
Function $\bsearch(T, b, e, c)$ takes SA interval $[b, e]$ of $Y$ and character 
$c\in \Sigma$ as input and returns SA interval $[b', e']$ of $cY$.
We compute backward search using $LF$, $\rank$, and $\select$ queries on $L$. 
The LF function receives the first and last occurrences of $c$ in $L[b..e]$ and returns 
the first and last positions of the SA interval of $cY$ because $L[i]$ represents the character preceding $F[i]$ on $T$ for an integer $i \in [2..n]$.
We can compute the first and last occurrences of $c$ using $\rank$ and $\select$ queries on $L$. 

Formally, let $x = \rank(L, c, b-1)$ and $y = \rank(L, c, e)$; $b' = \LF(\select(L, c, x+1))$ and $e' = \LF(\select(L, c, y))$ hold if the length of the SA interval for $cY$ is not zero. 

\subsubsection{Run-length encoding and RLBWT}
For a string $T$, \emph{Run-length encoding} $\rle(T)$ is a 
partition of $T$ into substrings $f_{1}, f_{2}, \ldots, f_{r}$ such that 
each $f_{i}$ is a maximal repetition of the same character 
in $T$. We call each $f_i$ a \emph{run}.

{\em The RLBWT} of $T$ is the BWT encoded by the run-length encoding, i.e., $\rle(L)$.
The RLBWT is stored in $O(r)$ space because 
each run in the RLBWT can be encoded into a pair of integers $c$ and $\ell$, 
where $c$ is the character and $\ell$ is the length of the run.
We call such a representation the \emph{compressed form} of the RLBWT.

\subsection{LZ77}
For a string $T$, 
LZ77 parsing~\cite{LZ76} of the reversed $T$ 
greedily partitions $T$ into substrings~(phrases) $f_{z}, f_{z-1}, \ldots, f_{1}$ 
in right-to-left order such that each phrase is either (i) copied from a subsequent substring in $T$~(target phrase) or 
(ii) an explicit character~(character phrase). 
We denote LZ77 phrases of the reversed $T$ as $\LZ(T^R)$.

Formally, let $i'$ be the ending position of $f_{i}$ for $i \in [1, z]$, 
i.e., $i' = |f_{i'-1} \cdots f_{1}|+1$. 
Then $f_{1} = T[n]$, and $f_{i}$ is $T[i']$ for $i \in [2, z]$ 
if $T[i']$ is a new character~(i.e., $\Occ(T[i'+1..], T[i']) = \emptyset$); 
otherwise, $f_{i}$ is the longest suffix $P$ of $T[..i']$, which 
has right occurrences in $T$~(i.e., $|P| = \max \{ \ell \mid \Occ(T[i'-\ell+2..], T[i'-\ell+1..i']) \neq \emptyset, \ell \in [1, i'] \}$).

We can store LZ77 phrases in $O(z)$ space because 
we encode a target phrase into the pair $\langle p, \ell \rangle$ of the right occurrence $p$ and length $\ell$ of the phrase.
We call such representation the \emph{compressed form} of LZ77. 
For example, let $T = c bbbb baba aba b a$. 
Then $\LZ(T^{R}) = c,bbbb,baba,aba,b,a$, 
and the compressed form of $\LZ(T^{R})$ is 
$c, \langle 3, 4 \rangle,\langle 11, 4 \rangle,\langle 12, 3 \rangle,b,a$. 

\section{Policriti and Prezza's conversion algorithm}\label{sec:overview}
\begin{algorithm}[t]
    \KwData{the RLBWT $L$ of $T$ and the position $y$ of $T[n]$ on $L$}
	\KwResult{$\LZ(T^{R})$}
  \SetKwProg{Fn}{Function}{}{}
  $(b, e, p, \ell, x) \leftarrow (1, n, -1, 1, n)$\tcc*{Initialization}
  \While{$x \geq 1$}{
  \tcc{Let P be T[x..x+l-1]}
    $c \leftarrow \access(L, y)$\tcc*{Access T[x]}
    $[b', e'] \leftarrow \bsearch(T, b, e, c)$\tcc*{Compute the SA interval of P}
    $p'  \leftarrow \any(\SA, b', e', x+1)$\;
    \If(\tcc*[f]{Any right occurrence of P was not found}){$p' = -1$}{
        \If{$\ell > 1$}{
         $\mathit{output}$ $\langle p, \ell-1 \rangle$\; 
        }\Else{
         $\mathit{output}$ $c$\;
         $(x,y)=(x-1, \LF(y))$\;
        }
      $(b, e, p, \ell) \leftarrow (1, n, p', 1)$\;
    }\Else{
      $(b, e, p, \ell, x, y) \leftarrow (b', e', p', \ell+1, x-1, \LF(y))$\;
    }

  }
  \caption{Policriti and Prezza's conversion algorithm.\label{algo:lzbwt}}
\end{algorithm}
Policriti and Prezza's conversion algorithm~\cite{DBLP:journals/algorithmica/PolicritiP18} converts 
a compressed string of $T$ in the RLBWT format to another compressed string of $T^R$ in the LZ77 format 
while using data structures in $O(r)$ space.
The data structures support four queries: backward search, the LF function, access queries on the RLBWT $L$, and RMTQ on the suffix array of $T$.
The RMTQ $\any(D, i, j, k)$ takes value $k$, interval $[b, e]$, and array $D$ as inputs and returns a value larger than $k$ on interval $[b, e]$ in $D$.

The algorithm extracts the original string from $L$ in right-to-left order using the LF function and access queries on $L$ 
and computes LZ77 phrases sequentially using backward search and RMTQ.
In each step, the algorithm extracts a suffix of the original string (i.e., current extracted string) then outputs the LZ77 phrase called {\em current pattern} in the suffix. 
In each step, the following two conditions for the current pattern are guaranteed: 
(i) the current pattern has at least one right occurrence or is the string of length $0$; 
(ii) the length of the current pattern is no less than that of the following current pattern (i.e., the next computed LZ77 phrase). 

For computing the current extracted string in each step, the algorithm computes
(i) the next character preceding the current extracted string, 
(ii) computes the SA interval corresponding to the current pattern using the backward search, and
(iii) finds any right occurrence of the current pattern in the SA interval using RMTQ. 
If such right occurrence of the extended pattern does not exist, the algorithm outputs the current pattern as the next LZ77 phrase.
When the length of the current pattern is $0$, the next phrase is the next character. 
The algorithm repeats the above step until it extracts the whole string and outputs LZ77 phrases of $T^R$. 
Algorithm~\ref{algo:lzbwt} shows a pseudo code of the algorithm. 

The algorithm uses two dynamic data structures: one for supporting backward search, 
the LF function, and access queries on $L$ in $O(\log r)$ time and $O(r)$ space; 
the other for supporting RMTQ in $O(\log r)$ time and $O(r)$ space, which is detailed in the next subsection.

\subsection{RMTQ data structure}\label{sec:right}
\begin{figure}[htbp]
  \centerline{
		\includegraphics[width=0.28\textwidth]{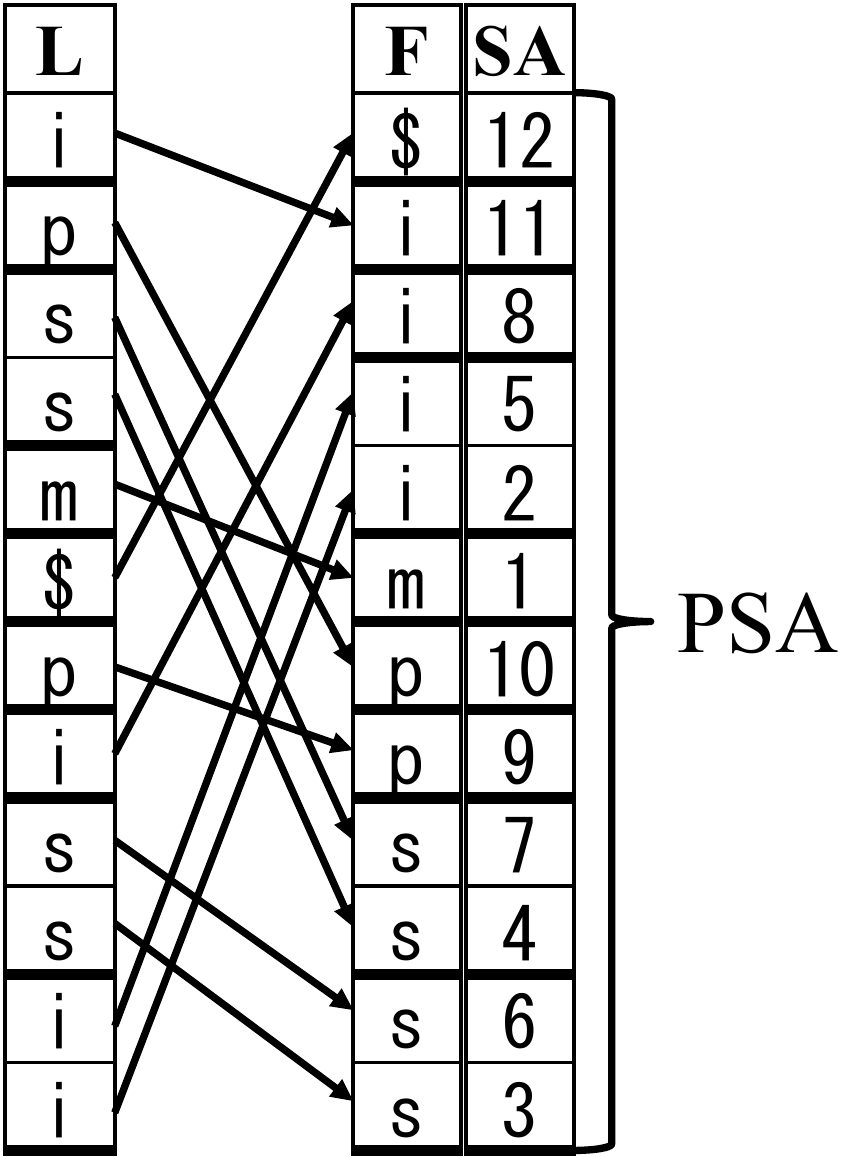} 
		\includegraphics[width=0.35\textwidth]{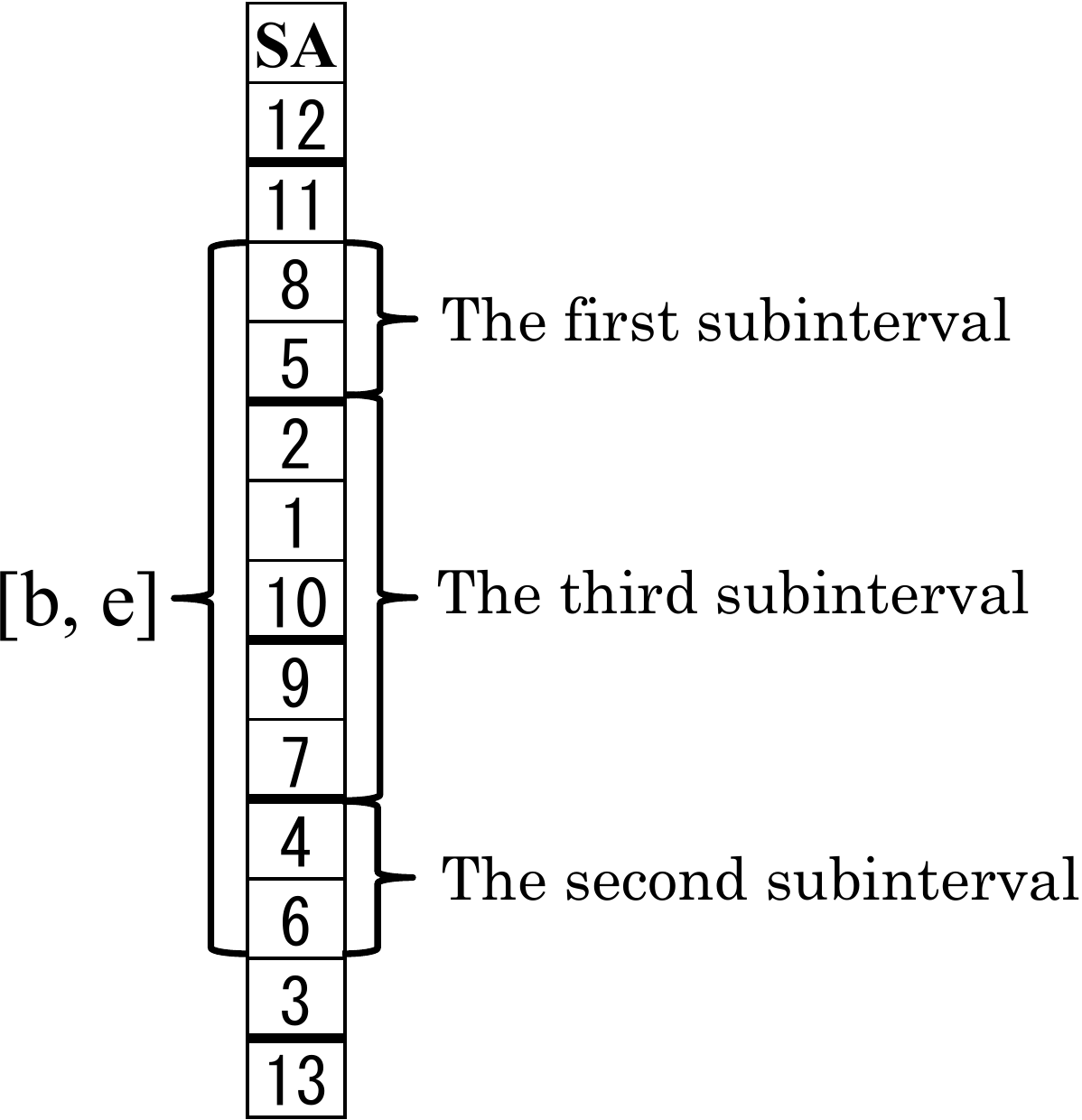} 
  }
	  \caption{ 
	  Left figure illustrates partitioned suffix array~(PSA). 	  
	  Bold horizontal lines on suffix array represent partitions on PSA; 
	  hence, PSA is $(12), (11), (8), (5, 2), (1), (10), (9), (7, 4), (6, 3)$. 
	  Right figure illustrates three subintervals used in Section~\ref{sec:improve_rmtq}. 
	  }
 \label{fig:psa}
\end{figure}
Policriti and Prezza presented an RMTQ data structure for fixed $k$, 
which can be updated when $k$ is decremented. 
The construction for RMTQ data structure partitions the suffix array of $T$ into subarrays 
for every run in $L$ by the LF function, resulting in $r$ subarrays in total. 
Since the $i$-th occurrence of any character $c$ in $F$ corresponds to the $i$-th occurrence of $c$ in $L$, 
the characters on every run in $L$ also occur continuously on $F$. 
The $F$ can be partitioned into $r$ substrings such that each substring corresponds to a run in $L$. 
We call such subarrays \emph{partitioned suffix arrays}~(PSAs).
The left in Figure~\ref{fig:psa} shows an example for the PSA of $T$ in Figure~\ref{fig:sa}. 

The data structure does not store the whole PSA. 
Instead, it stores only the first and last values larger than $k$ on each subarray of the PSA and their corresponding positions on it. 
The red-black tree is used to store those positions. 
We call such a first position (respectively, last position) on the $i$-th subarray for fixed $k$ 
a $k$-\emph{open position} (respectively, $k$-\emph{close position}) on the $i$-th subarray, which is denoted as $\fopen(i, k)$ (respectively, $\fclose(i, k)$). 

$\any (D, i, j, k)$ using the data structure is divided into two cases according to the relationship between the query interval $[b,e]$ and PSA: 
(A) there exists a subarray $\SA[p..q]$ of the PSA completely including interval $[b, e]$~(i.e., $p < b \leq e < q$ holds); and 
(B) such a subarray does not exist.
Both cases are detailed as follows.

For case (B), the query interval is partitioned into some subarrays of the PSA 
and contains either a prefix or suffix of each subarray.
Therefore, the interval contains at least one of the $k$-open and $k$-close positions 
if and only if the interval contains a value larger than $k$. 
We select an answer of RMTQ from the $k$-open positions and $k$-close positions on subarrays in the PSA in $O(\log r)$ time 
using the red-black tree storing the set of $k$-open and $k$-close positions. 

For case (A), the $\any (D,i,j,k)$ is computed using its computation result in the previous iteration of Algorithm~\ref{algo:lzbwt}.
The query interval $[b,e]$ represents the SA interval of a string $P$, 
and the length of $P$ is at least $2$ because the SA interval of a character does not satisfy case (A).
This means that the query interval in the previous iteration represents the SA interval of $P[2..]$, and
Algorithm~\ref{algo:lzbwt} computes the answer $p$~($\neq -1$) in the previous iteration. 
Thus, in case (A), $(p-1)$ is the answer for $P$ because
$P[1]$ is the character on $L$ such that $p$ is at the same position on the suffix array. 

Formally, 
let $\rle(L) = L_{1}, L_{2}, \ldots, L_{r}$, 
$p(i)$ be the starting position of $L_{i}$ in $L$~(i.e., $p(i) = |L_{1} \cdots L_{i}| - |L_{i}| + 1$), 
and $X$ be the permutation of $[1..r]$ such that 
$\LF(p(X[1])) < \ldots < \LF(p(X[r]))$ holds. 
Then the PSA of $T$ is $r$ subarrays $s_{1}, \ldots, s_{r}$ such that 
$s_{i} = \SA[\LF(p(X[i]))..$ $\LF(p(X[i]))+ |L_{X[i]}| -1]$ holds for all $i \in [1,r]$. 
Let $\fopen(i, k) = \min (\{ n + 1 \} \cup \{ j \mid \SA[j] \geq k, j \in [s_{i}..s_{i+1}-1] \} )$ and $\fclose(i, k) = \max (\{ 0 \} \cup \{ j \mid \SA[j] \geq k, j \in [s_{i}..s_{i+1}-1] \} )$ for $k \in [1, n]$ and $i \in [1, r]$.  
Let $\mathcal{T}_{k}$ be the set of $k$-open and close positions in the suffix array of $T$, 
i.e, $\mathcal{T}_{k} = \{ \fopen(1, k), \fclose(1, k) \ldots, \fopen(r, k), \fclose(r, k)  \}$. 
Then the following lemma holds. 
\begin{lemma}[\cite{DBLP:journals/algorithmica/PolicritiP18}]
Let $P$ be a substring of $T$ starting at a position $k$ and $\SA[b..e]$ be the SA interval of $P$. 
(1) In case (A), 
the length of $P$ is at least 2 and  
$\any(SA, b, e, k)$ can return $\any(SA, b', e', k+1)-1$, 
where $\SA[b..e]$ is the SA interval of $P[2..]$. 
(2)
In case (B), if $\mathcal{T}_{k} \cap [b..e] \neq \emptyset$ holds, 
then $\any(SA, b, e, k)$ can return the value at any position in $\mathcal{T}_{k} \cap [b..e]$; 
otherwise $\any(\SA, b, e, k) = -1$ holds. 
\end{lemma}

\section{Data structures for faster conversion}
\begin{algorithm}[t]
	\KwResult{$\any(\SA, b, e, x)$}
  \SetKwProg{Fn}{Function}{}{}
    $\hat{b} \leftarrow \pred(Z, b)$\tcc*{Get the index $\hat{b}$ of the subarray on the position b}
    $\hat{e} \leftarrow \pred(Z, e)$\tcc*{Get the index $\hat{e}$ of the subarray on the position e}
    $(x^{close}, v^{close}) \leftarrow M^{close}_{k}[\hat{b}]$\;
    $(x^{open}, v^{open}) \leftarrow M^{open}_{k}[\hat{e}]$\;
    $v^{\circ} \leftarrow \rmq(M, \hat{b}+1, \hat{e}-1)$\;
        \If{$x^{close} \in [b, e]$}{
            \Return $v^{close}$\;
        }\ElseIf{$x^{open} \in [b, e]$}{
            \Return $v^{open}$\;
        }\ElseIf{$v^{\circ} > k$ and $\hat{b} + 1 \leq \hat{e} - 1$}{
            \Return $v^{\circ}$\;
        }\Else{
            \Return $-1$\;
        }
  \caption{Our $\any(\SA, b, e, x)$ algorithm for case (B).\label{algo:lzbwt2}}
\end{algorithm}
We improve the query time in the data structure for backward search, the LF function, access query on $L$, and RMTQ for case (B) by presenting two novel data structures:
one supports the RMTQ for case (B); and the other supports backward search, the LF function, and access query on $L$.
Our data structures use static predecessor data structures in the internal and 
improve four query times by choosing the predecessor data structure with the fastest (estimated) query time. 
We show the results of our data structures, which are summarized in Table~\ref{table:result}.

\subsection{RMTQ data structure in case (B)}\label{sec:improve_rmtq}
Our RMTQ data structure for fixed $k$ consists of four arrays of length $r$: $k$-\emph{open array} $M_{k}^{open}$, 
$k$-\emph{close array} $M_{k}^{close}$, \emph{max value array} $M$, and \emph{starting position array} $Z$.
Data structures for the RMQ and predecessor query are built on $M$ and $Z$, respectively.

The $i$-th element of the $k$-open array~(respectively, $k$-close array) 
stores the pair of $k$-open position~(respectively, $k$-close position) and 
its corresponding value on the $i$-th subarray of the PSA.
The $k$-open and $k$-close arrays can be updated when $k$ is decremented. 
The $i$-th element of the max value array $M$ stores the maximum value on the $i$-th subarray of the PSA.
The $i$-th element of the starting position array stores the starting position of the $i$-th subarray on the PSA~(i.e., $Z[i] = p(X[i])$). 

Our algorithm for RMTQ~($\any(\SA, b, e, x)$ algorithm) consists of three parts:
(i) it divides a given query interval into at most three subintervals; 
(ii) computes the RMTQ for each subinterval; 
and (iii) returns the final answer of the RMTQ for a given interval using answers for subintervals.
The three subintervals are defined as follows: 
the first subinterval is on the first subarray of the PSA in the query interval,  
the second subinterval is on the last subarray of the PSA in the query interval, and 
the third subinterval is on the remaining subarrays. 
The right figure in Figure~\ref{fig:psa} illustrates those three subintervals. 
We compute the three subintervals using predecessor queries on $Z$ for a given query interval. 

The first subinterval is on a suffix of the first subarray; 
hence, the first subinterval has a value larger than $k$ if and only if the subinterval contains the $k$-close position of the first subarray.
Similarly, the second subinterval has a value larger than $k$ if and only if the subinterval contains the $k$-open position of the last subarray. 
The third subinterval is on middle subarrays; hence,
the third subinterval has a value larger than $k$ if and only if the maximal value is larger than $k$ on the subarrays. 
Therefore, we compute the RMTQ for the first subinterval~(respectively, the second subinterval)
by accessing the element of the first subarray on the $k$-close array~(respectively, the element of the last subarray on the $k$-open array).
We also compute the RMTQ for the third subinterval using the RMQ whose query interval covers the third subinterval on $M$. 

Algorithm~\ref{algo:lzbwt2} shows a pseudo code of our RMTQ algorithm. 
Since we can compute the RMQ in constant time~(e.g., \cite{DBLP:journals/siamcomp/FischerH11}), 
the query time of our RMTQ algorithm depends on the performance of predecessor queries on $Z$.
We can also convert $k$-open and close arrays to $(k-1)$-open and close arrays by changing at most two elements 
because the conversion can change the only element of the subarray containing $k$. 

Formally, let $\hat{b}$ and $\hat{e}$ be the ranks of the subarray of the PSA of $T$ which contains the positions $b$ and $e$, respectively, 
i.e., $\hat{b} = \pred(Z, b)$ and $\hat{e} = \pred(Z, e)$. 
For $k \in [1, n]$, let $M^{open}_{k}$ be the array of size $r$ such that for $i \in [1, r]$, 
$M^{open}_{k}[i] = (\SA[\fopen(i,k)], \fopen(i,k))$ if $\fopen(i, k) \neq n+1$ holds; otherwise $M^{open}_{k}[i] = (-1, \fopen(i,k))$. 
Similarly, let $M^{close}_{k}$ be the array of size $r$ such that for $i \in [1, r]$, 
$M^{close}_{k}[i] = (\SA[\fclose(i,k)], \fclose(i,k))$ if $\fclose(i, k) \neq 0$; otherwise 
$M^{close}_{k}[i] = (-1, \fclose(i,k))$. 
Let $M$ be the integer array of length $r$ such that $M[i]$ stores the maximal value in the $i$-th subarray of the PSA of $T$, 
i.e., $M[i] = \rmq(\SA, p(i), p(i+1)-1)$ for all $i \in [1, r]$. 
Let $Q(m, u)$ be the query time of the predecessor query on a set $S$ of size $m$ from a universe $[1, u]$ 
by a predecessor data structure of $O(m)$ space. 
Then the following lemmas hold. 
\begin{lemma}
In case (B), $\any(\SA, b, e, k) \neq -1$ holds if and only if there exists an answer of $\any(\SA, b, e, k)$ 
in $M^{open}_{k}[\hat{b}], M^{close}_{k}[\hat{e}]$, or $\rmq(M, \hat{b}+1, \hat{e}-1, k)$. 
\end{lemma}
\begin{lemma}\label{lem:caseBquerytime}
Our data structure for case (B) can compute $\any(SA, b, e, k)$ in $O(1 + Q(r, n))$ time. 
The space usage is $O(r)$ space.
\end{lemma}
\begin{proof}
We construct the predecessor data structure for $Z$, which supports predecessor queries in $Q(r, n)$ time, 
and the RMQ data structure for $M$ which supports the RMQ in $O(1)$ time using~\cite{DBLP:journals/siamcomp/FischerH11}. 
Then Lemma~\ref{lem:caseBquerytime} holds by Algorithm~\ref{algo:lzbwt2}.
\end{proof}

\begin{lemma}\label{lem:updatemk}
For given position $x$ of a $k$ on the suffix array of $T$~(i.e., $\SA[x] = k$), 
we can convert $M^{open}_{k}$ and $M^{close}_{k}$ to $M^{open}_{k-1}$ and $M^{close}_{k-1}$ in $O(1+Q(r, n))$ time 
using the predecessor data structure for $Z$.
\end{lemma}
\begin{proof}
$M^{open}_{k}[i] = M^{open}_{k-1}[i]$ and $M^{close}_{k}[i] = M^{close}_{k-1}[i]$ hold for all $i \in ([1, r] \setminus \{ p \})$, 
where $p = \pred(Z, x)$. Therefore, we appropriately update $M^{open}_{k}[p]$ and $M^{open}_{k}[p]$ for $k-1$.
\end{proof}

\subsection{Data structure for backward search, LF, and access queries}
We leverage the static data structures presented by Gagie et al.~\cite{DBLP:conf/soda/GagieNP18} for backward search, the LF function, and access queries on $L$ instead of Policriti and Prezza's dynamic data structure.
The static data structures compute three queries by executing only the constant number of predecessor queries, 
and the three queries using the static data structures can be faster than those using Policriti and Prezza's dynamic data structure.
 
Since the time for the predecessor query on the static data structures is proportional to the alphabet size of the input RLBWT, the alphabet size slightly increases the query times. 
For faster queries, 
we replace one of the static data structures for the predecessor queries depending on the alphabet size with an array of size $\sigma$. 
We obtain the following lemma. 
\begin{lemma}\label{lem:rlbwtdata}
Let $C(m, u)$ be the construction time for the predecessor data structure of $O(m)$ space, which supports predecessor queries in $Q(m, u)$ time. 
We can construct the data structure of $O(r + \sigma)$ space, 
which supports $\bsearch$, $\LF$, and $\access$ queries for $L$ in $O(1 + Q(r, n))$ time, by processing the RLBWT of $T$ in
$O(C(r, n) + \sigma + r)$ time and $O(r + \sigma)$ working space. 

\end{lemma}
\begin{proof}
See Appendix.
\end{proof}

\subsection{Improved Policriti and Prezza's algorithm}\label{sec:analysis}
Algorithm~\ref{algo:lzbwt} using our data structures requires $O(r + \sigma)$ space,
which can be $\omega(r)$ when $\sigma \geq r$, e.g., $\sigma = n^2$. 
To bound the space usage to $O(r)$, we reduce the alphabet size of the RLBWT $L$ to at most $r$ by renumbering characters in $L$.

We modify Algorithm~\ref{algo:lzbwt} as follows: 
(i) We replace each character $c$ in $L$ with the rank of $c$ in $L$~(i.e., $|\{L[i] \mid L[i] \leq c, i \in [1, n]\}|$) 
and construct the new RLBWT $L'$ over the alphabet of at most $r$.
We call the converted string the \emph{shrunk string} of $L$.
(ii) We convert $L'$ to LZ77 phrases of the string $T'^{R}$ recovered from $L'$ using Algorithm~\ref{algo:lzbwt}. 
(iii) We recover the LZ77 phrases of $T^{R}$ from that of $T'^{R}$ using the \emph{inverse array} $W$,  
where the $i$-th element of the array stores the original character of rank $i$~(i.e., $W[L'[i]] = L[i]$ holds for any $i \in [1, n]$). 
The modified algorithm works correctly because 
(1) our backward search queries receive only characters that appear in the RLBWT,
(2) $L'$ is the RLBWT of the shrunk string of $T$, and
(3) the form of LZ77 phrases is independent of the alphabet, i.e., 
we obtain the $i$-th LZ77 phrase of $T^{R}$ by mapping characters in the $i$-th LZ77 phrase of $T'^{R}$ into the original characters.

Finally, we obtain a conversion algorithm from RLBWT to LZ77 in $O(n(1 + Q(r, n)) + C(r, n))$ time and $O(r)$ space, and
the algorithm depends on the performance of the static predecessor data structure.
There exist two predecessor data structures such that (1) $Q(r, n) = O(\sqrt{\log r / \log \log r})$ 
and $C(r, n) = O(r \sqrt{\log r / \log \log r})$ hold~\cite{DBLP:journals/jcss/BeameF02} 
and (2) $Q(r, n) = O(\log \log n)$ and $C(r, n) = O(r)$ hold~\cite{DBLP:conf/cpm/0001G15}. 
Since we can compute $r$ and $n$ by processing the RLBWT in $O(r)$ time, 
we choose the faster predecessor data structure between those predecessor data structures. 
Therefore, we obtain our results of our data structures, which are listed in Table~\ref{table:result}. 

Formally, the following lemmas and theorem hold.
\begin{lemma}\label{lem:effective2}
The following statements hold.
(1) We can compute the shrunk string $L'$ of $L$ and the inverse array $W$ in $O(r)$ time and working space.
(2) The $L'$ is the RLBWT of the shrunk string $T'$ of $T$.
(3) The $|\LZ(T'^{R})| = |\LZ(T^{R})|$ and $\LZ(T^{R})[i][j] = W[\LZ(T'^{R})[i][j]]$ hold for $i \in [1, z]$ and $j \in [1, |\LZ(T^{R})[i]|]$,
where $z$ is the number of LZ77 phrase of $T^{R}$.
(4) We can convert the compressed form of $\LZ(T'^{R})[i]$ to that of $\LZ(T^{R})[i]$ in constant time using $W$ for all $i \in [1, z]$.
\end{lemma}
\begin{proof}
(1)
We construct the string $E$ that consists of $r$ first characters 
in the run-length encoding of $L$~(i.e., $E = \rle(L)[1][1], \rle(L)[2][1], \ldots, \rle(L)[r][1]$) and 
construct the suffix array of $E$ in $O(r)$ time and working space~\cite{DBLP:conf/icalp/KarkkainenS03}. 
We construct the shrunk string $E'$ of $E$ and $W$ using the suffix array and construct $L'$ using $E'$.
(2) Let $\SA$ and $\SA'$ be the suffix arrays of $T$ and $T'$. Then $\SA[i] = \SA'[i]$ holds for all $i \in [1, n]$. 
Therefore, the RLBWT of $T'$ is $L'$. 
(3) This holds because $\Occ(T, T[x..y]) = \Occ(T',T'[x..y])$ holds for two integers $1 \leq x \leq y \leq n$. 
(4) If $\LZ(T'^{R})[i]$ is a target phrase, we return the phrase as $\LZ(T^{R})[i]$. 
Otherwise, we return $W[\LZ(T'^{R})[i]]$ as $\LZ(T^{R})[i]$.
\end{proof}
\begin{lemma}\label{lem:caseBconstructtime}
We can construct the data structure of Lemma~\ref{lem:caseBquerytime} in $O(n(1+ Q(r, n)) + C(r, n) + \sigma)$ time and 
$O(r + \sigma)$ working space using the RLBWT data structure of Lemma~\ref{lem:rlbwtdata}. 
\end{lemma}
\begin{proof}
See Appendix.
\end{proof}
\begin{theorem}\label{theo:main}
There exists a conversion algorithm from RLBWT to LZ77 in 
$O(n(1 + Q(r, n)) + C(r, n))$ time and $O(r)$ space.
\end{theorem}
\begin{proof}
We already have described our algorithm in Section~\ref{sec:analysis}. 
Each step of Algorithm~\ref{algo:lzbwt} additionally needs to update $M^{open}_{k}$, $M^{close}_{k}$ and determine either case (A) or (B) for a given query interval.
The total time is $O(1 + Q(r, n))$ using predecessor queries on $L$ and Lemma~\ref{lem:updatemk}.
Therefore, Theorem~\ref{theo:main} holds by Lemmas~\ref{lem:caseBquerytime},~\ref{lem:rlbwtdata},~\ref{lem:effective2}, and~\ref{lem:caseBconstructtime}. 
\end{proof}

\section{Conclusion}
We presented a new conversion algorithm from RLBWT to LZ77 in $O(n(1 + Q(r, n)) + C(r, n))$ time and $O(r)$ space. 
By leveraging the fastest static predecessor data structure using $O(r)$ space, 
we obtain the conversion algorithm that runs in $O(n \min \{ \log \log n, \sqrt{\frac{\log r}{\log\log r}} \})$ time.
This result improves the previous result in $O(n \log r)$ time and $O(r)$ space. 

We have the following open problem: 
can we archive the conversion from RLBWT to LZ77 in $O(n)$ time and $O(r)$ space? 
It is difficult to archive the $O(n)$ time complexity with our approach because 
any predecessor data structure for a set using $m^{O(1)}$ words of $(\log |U|)^{O(1)}$ bits requires $\Omega(\sqrt{\log m / \log \log m})$ query time in the worst case~\cite{DBLP:journals/jcss/BeameF02},
where $m$ is the number of elements in the set and $U$ is the universe of elements. 
Kempa's conversion algorithm can run faster than our algorithm, 
but it requires $\omega(r)$ working space in the worst case.
Thus, a new approach is required to solve this open problem.


\bibliography{ref}

\begin{thebibliography}{10}

\bibitem{DBLP:conf/cpm/BannaiGIT13}
Hideo Bannai, Pawel Gawrychowski, Shunsuke Inenaga, and Masayuki Takeda.
\newblock Converting {SLP} to {LZ78} in almost linear time.
\newblock In {\em Proceedings of {CPM}}, pages 38--49, 2013.

\bibitem{DBLP:conf/spire/BannaiIT12}
Hideo Bannai, Shunsuke Inenaga, and Masayuki Takeda.
\newblock Efficient {LZ78} factorization of grammar compressed text.
\newblock In {\em Proceedings of {SPIRE}}, pages 86--98, 2012.

\bibitem{DBLP:journals/jcss/BeameF02}
Paul Beame and Faith~E. Fich.
\newblock Optimal bounds for the predecessor problem and related problems.
\newblock {\em J. Comput. Syst. Sci.}, 65(1):38--72, 2002.

\bibitem{DBLP:journals/algorithmica/BilleCCGSVV18}
Philip Bille, Anders~Roy Christiansen, Patrick~Hagge Cording, Inge~Li G{\o}rtz,
  Frederik~Rye Skjoldjensen, Hjalte~Wedel Vildh{\o}j, and S{\o}ren Vind.
\newblock Dynamic relative compression, dynamic partial sums, and substring
  concatenation.
\newblock {\em Algorithmica}, 80(11):3207--3224, 2018.

\bibitem{burrows1994block}
Michael Burrows and David~J Wheeler.
\newblock A block-sorting lossless data compression algorithm.
\newblock {\em Technical report}, 1994.

\bibitem{DBLP:conf/cpm/0001G15}
Johannes Fischer and Pawel Gawrychowski.
\newblock Alphabet-dependent string searching with wexponential search trees.
\newblock In {\em Proceedings of {CPM}}, pages 160--171, 2015.

\bibitem{DBLP:journals/siamcomp/FischerH11}
Johannes Fischer and Volker Heun.
\newblock Space-efficient preprocessing schemes for range minimum queries on
  static arrays.
\newblock {\em {SIAM} J. Comput.}, 40(2):465--492, 2011.

\bibitem{DBLP:conf/soda/GagieNP18}
Travis Gagie, Gonzalo Navarro, and Nicola Prezza.
\newblock Optimal-time text indexing in {BWT}-runs bounded space.
\newblock In {\em Proceedings of {SODA}}, pages 1459--1477, 2018.

\bibitem{DBLP:journals/tcs/HonSS11}
Wing{-}Kai Hon, Kunihiko Sadakane, and Wing{-}Kin Sung.
\newblock Succinct data structures for searchable partial sums with optimal
  worst-case performance.
\newblock {\em Theor. Comput. Sci.}, 412(39):5176--5186, 2011.

\bibitem{DBLP:journals/tcs/Jez16}
Artur Jez.
\newblock A really simple approximation of smallest grammar.
\newblock {\em Theor. Comput. Sci.}, 616:141--150, 2016.

\bibitem{DBLP:conf/icalp/KarkkainenS03}
Juha K{\"{a}}rkk{\"{a}}inen and Peter Sanders.
\newblock Simple linear work suffix array construction.
\newblock In {\em Proceedings of {ICALP}}, pages 943--955, 2003.

\bibitem{DBLP:conf/soda/Kempa19}
Dominik Kempa.
\newblock Optimal construction of compressed indexes for highly repetitive
  texts.
\newblock In {\em Proceedings of {SODA}}, pages 1344--1357, 2019.

\bibitem{LZ76}
Abraham Lempel and Jacob Ziv.
\newblock On the complexity of finite sequences.
\newblock {\em IEEE Transactions on information theory}, 22(1):75--81, 1976.

\bibitem{DBLP:journals/siamcomp/ManberM93}
Udi Manber and Eugene~W. Myers.
\newblock Suffix arrays: {A} new method for on-line string searches.
\newblock {\em {SIAM} J. Comput.}, 22(5):935--948, 1993.

\bibitem{DBLP:journals/algorithmica/PolicritiP18}
Alberto Policriti and Nicola Prezza.
\newblock {LZ77} computation based on the run-length encoded {BWT}.
\newblock {\em Algorithmica}, 80(7):1986--2011, 2018.

\bibitem{DBLP:journals/tcs/Rytter03}
Wojciech Rytter.
\newblock Application of lempel-ziv factorization to the approximation of
  grammar-based compression.
\newblock {\em Theor. Comput. Sci.}, 302(1-3):211--222, 2003.

\bibitem{DBLP:journals/corr/abs-1811-01472}
Kensuke Sakai, Tatsuya Ohno, Keisuke Goto, Yoshimasa Takabatake, Tomohiro I,
  and Hiroshi Sakamoto.
\newblock Repair in compressed space and time.
\newblock {\em CoRR}, abs/1811.01472, 2018.

\end{thebibliography}

\clearpage
\section*{Appendix A: The proof of Lemma~\ref{lem:rlbwtdata}}
We can compute $\LF$ and $\bsearch$ queries using $C$, rank, select, access queries for $L$ 
and construct $C$ by processing the RLBWT of $T$ in $O(r + \sigma)$ time and working space.
Therefore, we give the data structures for rank, select, and access queries for $L$ by the following lemmas. 
\begin{lemma}\label{lem:rlbwtdata1}
We can construct the data structure of $O(r + \sigma)$ space 
that supports $\rank$ queries for $L$ in $O(1+Q(r, n))$ time by processing the RLBWT of $T$ in
$O(C(r, n) + \sigma + r)$ time and $O(r + \sigma)$ working space. 
\end{lemma}
\begin{proof}
We compute rank queries for a character $c$ by the constant number of accessing elements on two arrays $B_{c}$ and $V_{c}$ and the constant number of predecessor queries on $B_{c}$.
Array $B_{c}$ is the array storing sorted starting positions of runs of $c$ in $L$,
and $V_{c}$ is the integer array such that $V_{c}[i]$ stores the rank of the first character of the $i$-th run of character $c$.
We can construct $B_{1}, \ldots, B_{\sigma}$ and $V_{1}, \ldots, V_{\sigma}$ in $O(r + \sigma)$ time and working space by processing the RLBWT of $T$ 
and predecessor data structures for rank queries in $O(1 + C(r, n))$ time. 
Therefore Lemma~\ref{lem:rlbwtdata1} holds. 

Formally, let $\run(c)$ be the number of runs of character $c$ in $L$~(i.e., $\run(c) = |\{ i \mid L_{i}[1]=c , i \in [1, r] \}|$). 
Array $B_{c}[i]$ stores the starting position of the $i$-th run of character $c$ for all $c \in \Sigma$ and $i \in [1, \run(c)]$, 
$V_{c}[i] = \rank(L, c, B_{c}[i])$ for $c \in \Sigma$ and $i \in [1, \run(c)]$ and $V_{c}[\run(c)+1] = \rank(L, c, n) + 1$. 
If $L[x] = c$ holds, then $\rank(L, c, x) = V_{c}[t] + x - B_{c}[t]$ holds; otherwise, $\rank(L, c, x) = V_{c}[t+1] - 1$ holds,
where $t = \pred(B_{c}, x)$.
Since $V_{c}[i+1] - V_{c}[i]$ represents the length of the $i$-th run of character $c$, 
we can compute $L[x] = c$ using predecessor queries, 
i.e., $x-B_{c}[t]+1 \leq \ell$ holds if and only if $L[x] = c$ holds, 
where $\ell = V_{c}[\pred(B_{c}, x)+1] - V_{c}[\pred(B_{c}, x)]$. 
\end{proof}
\begin{lemma}\label{lem:rlbwtdata2}
We can construct the data structure of $O(r + \sigma)$ space 
that supports $\select$ queries for $L$ in $O(1+Q(r, n))$ time by processing the RLBWT of $T$ in 
$O(C(r, n) + \sigma + r)$ time and $O(r + \sigma)$ working space. 
\end{lemma}
\begin{proof}
We compute select queries for a character $c$ using three arrays $C$, $V_{c}$, and $B_{c}$ and the predecessor on $V_{c}$.
When $\select(L, c, x) \neq n+1$ holds, 
$\select(L, c, x) = B_{c}[\pred(V_{c}, x)] + x - V_{c}[\pred(V_{c}, x)]$ holds. 
Since $C[c+1]-C[c]$ represents the number of $c$s in $L$,
we can compute $\select(L, c, x) \neq n+1$ using $C$. 
We can construct $B_{1}, \ldots, B_{\sigma}$ and $V_{1}, \ldots, V_{\sigma}$ in $O(r + \sigma)$ time and working space by processing the RLBWT of $T$ 
and predecessor data structures for select queries in $O(1 + C(r, n))$ time. 
Therefore Lemma~\ref{lem:rlbwtdata2} holds. 
\end{proof}
\begin{lemma}\label{lem:rlbwtdata3}
We can construct the data structure of $O(r)$ space 
that supports access queries for $L$ in $O(1+Q(r, n))$ time by processing the RLBWT of $T$ in 
$O(C(r, n) + r)$ time and $O(r)$ working space. 
\end{lemma}
\begin{proof}
We compute access queries using two arrays $B$ and $E$ and the predecessor on $B$. 
Array $B$ is the sorted starting positions of runs in $L$~(i.e., $B = p(1), p(2), \ldots, p(r)$), 
and $E$ is the first characters of runs in $L$~(i.e., $E = \rle(L)[1][1], \rle(L)[2][1], \ldots, \rle(L)[r][1]$). 
We can construct $B$ and $E$ in $O(r)$ time by processing the RLBWT. 
Since we can compute $L[i]$ by $E[\pred(B, i)]$ for a given integer $i$,  
Lemma~\ref{lem:rlbwtdata3} holds. 
\end{proof}

Therefore Lemma~\ref{lem:rlbwtdata} holds by Lemmas~\ref{lem:rlbwtdata1},~\ref{lem:rlbwtdata2}, and~\ref{lem:rlbwtdata3}.

\section*{Appendix B: The proof of Lemma~\ref{lem:caseBconstructtime}}
\begin{proof}
Our data structure for RMTQ consists of $M, M_{n}^{open}, M_{n}^{close}, Z$, 
the RMQ data structure for $M$, and the predecessor data structure for $Z$.
We construct $Z$ using $c$-\emph{integer sequence}. 
The $c$-integer sequence is the subarray of $Z$ such that subarray $Z[b..e]$
is on the run of a character $c$ in $F$, i.e., 
$b = \min \{ X[u] \mid L[p(u)] = c, u \in [1, r] \}$ and $e = \max \{ X[u] \mid L[p(u)] = c, u \in [1, r] \}$.
We construct $1$-integer sequence, $\ldots$, $\sigma$-integer sequence in $O(r (1+Q(r,n)) + \sigma)$ 
using the LF function since $\LF(p(i)) < \LF(p(j))$ holds for two positions $i$ and $j$
such that $L[i] = L[j]$ and $i < j$ hold. 
We obtain $Z$ by concatenating the sequences.
The total time is $O(r (1+Q(r,n)) + \sigma)$ and the working space is $O(r + \sigma)$. 

We construct $M$ using the $\LF$ function and predecessor on $Z$ in $O(1 + Q(r, n))$ time 
since $\LF(i)$ represents a position on the suffix array of $T$ for $i \in [1, n]$. 
We construct $M^{open}_{n}$ and $M^{close}_{n}$ in $O(r)$ time since 
$M^{open}_{n} = (-1, n+1), \ldots, (-1, n+1)$ and $M^{close}_{n} = (-1, 0), \ldots, (-1, 0)$. 
We construct the RMQ data structure for $M$ in $O(r)$ time and working space using~\cite{DBLP:journals/siamcomp/FischerH11}. 
Therefore Lemma~\ref{lem:caseBconstructtime} holds since $r \leq n$.
\end{proof}

\end{document}